\documentclass[draft,onecolumn]{IEEEtran}
\IEEEoverridecommandlockouts
\usepackage{cite}
\usepackage{amsmath,amssymb,amsfonts,amsthm}
\usepackage{algorithmic}
\usepackage{graphicx}
\usepackage{textcomp}
\usepackage{xcolor}
\def\BibTeX{{\rm B\kern-.05em{\sc i\kern-.025em b}\kern-.08em
    T\kern-.1667em\lower.7ex\hbox{E}\kern-.125emX}}
\newtheorem{thm}{Theorem}
\newtheorem{defn}{Definition}

\newtheorem{lem}{Lemma}

\newtheorem{corollary}{Corollary}

\usepackage{geometry}
\usepackage{setspace}
\begin{document}

\title{  Uncertainty principles for the short-time linear canonical transform of complex signals}
\author{\IEEEauthorblockN{Wen-Biao, Gao$^{1}$,
		Bing-Zhao, Li$^{1,2,\star}$,}\\
	\IEEEauthorblockA{
		1. School of Mathematics and Statistics, Beijing Institute of Technology, Beijing 102488, P.R. China}\\
	2. Beijing Key Laboratory on MCAACI, Beijing Institute of Technology, Beijing 102488, P.R. China\\
$^{\star}$Corresponding author: li$\_$bingzhao@bit.edu.cn
}

\maketitle

\begin{abstract}
\label{abstract}
\begin{spacing}{1.25}
The short-time linear canonical transform (STLCT) can be identified as a generalization of the short-time Fourier transform (STFT). It is a novel time-frequency analysis tool. In this paper, we generalize some different uncertainty principles for the STLCT of complex signals. Firstly, one uncertainty principle for the STLCT of complex signals in time and frequency domains is derived. Secondly, the uncertainty principle for the STLCT of complex signals in two STLCT domains is obtained.
Finally, the uncertainty principle for the two conditional standard deviations of the spectrogram associated with the STLCT is discussed.
\end{spacing}
\end{abstract}


\noindent  {\bf 1. Introduction.}

The linear canonical transform(LCT)  \cite{12,13,14,15,16,20,21} is an integral transform with four parameters (a,b,c,d) and it has many special cases such as the Fourier transform (FT), the fractional Fourier transform (FRFT). Three free parameters of the LCT make it more flexible.
The LCT has become a powerful analyzing tool and it has been widely used in some fields, such as applied mathematics, signal
processing, pattern recognition, optical system analysis and filter design \cite{13,14,15}.

Kou and Xu \cite{4} first proposed the short-time linear canonical transform (STLCT). The STLCT is also known
as the windowed linear canonical transform.
It modifies the traditional the linear canonical transform  (LCT) by utilizing a window. Then many scholars has also pay attention to the STLCT. Bahri and Ashino \cite{17} discussed some properties of short-time linear canonical transform. Zhang \cite{12} obtained a sampling theorem for the STLCT by means of a generalized Zak transform associated with the LCT. Zhang \cite{11} discussed the discrete windowed linear canonical transform.
Some scholars \cite{17,18} discussed the uncertainty principles associated with the STLCT. However, the uncertainty principles for the short-time linear canonical transform of complex signal moments have not been found. Therefore, based on the definition of the STLCT, we seek the result via mathematical proof and obtain that the uncertainty principles for the short-time linear canonical transform of complex signals.

The paper is organized as follows: In Section 2, we present some general
definitions. Some different uncertainty principles associated with the STLCT are provided in Section 3.
In Section 4, some conclusions are drawn.

\noindent  {\bf 2. Preliminary.}

Some relevant contents are presented in this section.
\begin{defn}(LCT)
Let $A=(a,b,c,d)$ be a matrix parameter satisfying $a,b,c,d\in \mathbb{R}$ and $ad-bc=1$. For a complex signal $f(t)=x(t)e^{i \phi(t)}\in L^{2}(\mathbb{R})$, where $x(t)$ and $\phi(t)$ are real-valued and differentiable functions, the LCT of $f(t)$ is defined as \cite{1,2}
\begin{align}
		\begin{split}
			L_{A}^{f}(u)=L_{A}[f(t)](u)=\begin{cases}
		\int_{-\infty}^{+\infty}f(t)K_{A}(t,u)\rm{d}\textit{t},   &b\neq0  \\
		\sqrt{d}e^{i\frac{cd}{2}u^{2}}f(du)),    &b=0
		\end{cases}
		\end{split}
	\end{align}
where the transform kernel of LCT is given by
\begin{align}
	K_{A}(t,u)=\frac{1}{\sqrt{i2\pi b}}e^{i\frac{a}{2b}t^{2}-i\frac{1}{b}tu+i\frac{d}{2b}u^{2}}
\end{align} \end{defn}
The following properties \cite{3} which will be useful in this paper:
\begin{align}
	K_{A}^{*}(t,u)=K_{A^{-1}}(u,t)
\end{align}
\begin{align}
	\int _{-\infty}^{+\infty}e^{\pm iux}\rm{d}\textit{u}=2\pi \delta(\textit{x}); \ \ \ \int _{-\infty}^{+\infty}\delta(t-u)\emph{f}(t)\rm{d}\textit{t}=\textit{f}(\textit{u})
\end{align}
\begin{align}
		\begin{split}
		\|f\|^{2}_{L^{2}(\mathbb{R})}=\|f\|^{2}=\int _{-\infty}^{+\infty}|f(t)|^{2}\rm{d}\textit{t}
		\end{split}
	\end{align}
where $A^{-1}=(d,-b,-c,a)$. In the following, we only consider $b\neq0$, since the LCT with $b=0$ is a chirp multiplication operation.
The inverse LCT expression is \cite{2,3}
\begin{align}
	f(t)=\int _{-\infty}^{+\infty}L_{A}^{f}(u)K_{A^{-1}}(u,t)\rm{d}\textit{u}
\end{align}

\begin{defn} (STLCT) Let $g\in L^2(\mathbb{R})$ be a window function. $A=(a,b,c,d)$ be a matrix parameter satisfying $a,b,c,d\in \mathbb{R}$ and $ad-bc=1$. The STLCT of a complex signal $f(\tau)=x(\tau)e^{i \phi(\tau)}\in L^{2}(\mathbb{R})$ with respect to $g$ is defined by \cite{4}
\begin{align}
	S^{A}_{g}f(t,u)=\int _{-\infty}^{+\infty}f_{t}(\tau) K_{A}(\tau,u)\rm{d}\tau
	\end{align}	
where $f_{t}(\tau)=f(\tau) g^{*}(\tau-t)$ and $K_{A}(\tau,u)$ is given by Eq.(2).
\end{defn}

$f_{t}(\tau)=f(\tau) g^{*}(\tau-t)$ is called the local signal, it is the signal we see at time $t$.
A different signal can be seen at a different time.

Let
\begin{align}
Q(t)=\int _{-\infty}^{+\infty}|f(\tau) g^{*}(\tau-t)|^{2}\rm{d}\tau
\end{align}	
then the local normalized signal is
\begin{align}
q_{t}(\tau)=\frac{1}{\sqrt{Q(t)}}f(\tau) g^{*}(\tau-t)
\end{align}	
so we can have $\int _{-\infty}^{+\infty}|q_{t}(\tau)|^{2}\rm{d}\tau=1$.

Inspired by Leon Cohen \cite{7}, we can define the local spectrum by
\begin{align}
L_{u}(\mu)=L_{A}^{f}(\mu)L_{A}^{g}(u-\mu)
\end{align}	
where $L_{A}^{f}(\mu), L_{A}^{g}(\mu)$ are given by Definition 1, $\mu$ is the running frequency and $u$ is the fixed frequency of interest. In order to study the time
properties of a signal at a particular frequency, we can window the spectrum of the signal, $L_{A}^{f}(\mu)$ , with a frequency window function, $L_{A}^{g}(\mu)$.

Let
\begin{align}
P(u)=\int _{-\infty}^{+\infty}|L_{A}^{f}(\mu)L_{A}^{g}(u-\mu)|^{2}\rm{d}\mu
\end{align}
then the normalized local spectrum is
\begin{align}
p_{u}(\mu)=\frac{1}{\sqrt{P(u)}}L_{A}^{f}(\mu)L_{A}^{g}(u-\mu)
\end{align}
The short-frequency time transform is
\begin{align}
s^{A}_{g}f(t,u)=\int _{-\infty}^{+\infty}L_{A}^{f}(\mu)L_{A}^{g}(u-\mu) K_{A}^{*}(t,\mu)\rm{d}\mu
\end{align}
So we can have
\begin{align}
|S^{A}_{g}f(t,u)|^{2}=|s^{A}_{g}f(t,u)|^{2}
\end{align}

Some series of important properties of the STLCT are listed in \cite{4}.

Now, we give a lemma, it is very important in this paper.
\begin{lem} Let $g\in L^2(\mathbb{R})$ be a window function. For a complex signal $f(\tau)=x(\tau)e^{i \phi(\tau)}\in L^{2}(\mathbb{R})$, we have
\begin{align}
	S^{A}_{g}f(t,u)=\int _{-\infty}^{+\infty}L^{\textit{f}}_{A}(\xi)G^{*}(\xi|u,t)\rm{d}\xi
	\end{align}	
where $G^{*}(\xi|u,t)=\sqrt{-i2\pi b}e^{i\frac{d'}{2b}(\xi-u)^{2}}L_{A_{1}}^{g}(\xi-u)^{*}K_{A}(t,u)K^{*}_{A}(t,\xi)$, $L_{A_{1}}^{g}(\xi-u)$ is the LCT of the window function $g(\tau)$, $A_{1}=(0,b,-\frac{1}{b},d')$, the parameter $b$ in the matrix
$A=(a,b,c,d)$ and $d'\in \mathbb{R}$.
\end{lem}
\begin{proof}
By Definition 2 and  Eq.(6), we have
\begin{align}
\begin{split}
	S^{A}_{g}f(t,u)&=\int _{-\infty}^{+\infty}f(\tau) \overline{g(\tau-t)} K_{A}(\tau,u)\rm{d}\tau\\
&=\int _{-\infty}^{+\infty} L_{A}^{f}(\xi)\int _{-\infty}^{+\infty}K_{A^{-1}}(\xi,\tau) \overline{g(\tau-t)}K_{A}(\tau,u)\rm{d}\tau\rm{d}\xi
\end{split}
	\end{align}	
Let $G^{*}(\xi|u,t)=\int _{-\infty}^{+\infty}K_{A^{-1}}(\xi,\tau) \overline{g(\tau-t)}K_{A}(\tau,u)\rm{d}\tau$ and $\tau-t=\eta$, then
\begin{align}
\begin{split}
	G^{*}(\xi|u,t)&=\int _{-\infty}^{+\infty}K_{A^{-1}}(\xi,\tau) \overline{g(\tau-t)}K_{A}(\tau,u)\rm{d}\tau\\
&=\int _{-\infty}^{+\infty}\overline{g(\eta)}\frac{1}{\sqrt{-i2\pi b}}\frac{1}{\sqrt{i2\pi b}}e^{-i\frac{(u-\xi)}{b}(\eta+t)+i\frac{d}{2b}(u^{2}-\xi^{2})}\rm{d}\eta\\
&=\frac{1}{\sqrt{i2\pi b}}\int _{-\infty}^{+\infty}\frac{1}{\sqrt{-i2\pi b}}\overline{g(\eta)}e^{i\frac{0}{-2b}\eta^{2}-i\frac{(\xi-u)}{-b}\eta+i\frac{d'}{-2b}(\xi-u)^{2}}\rm{d}\eta\\
&\times e^{i\frac{d'}{2b}(\xi-u)^{2}+i\frac{d}{2b}(u^{2}-\xi^{2})-i\frac{(u-\xi)}{b}t}\\
&=\frac{1}{\sqrt{i2\pi b}}e^{i\frac{d'}{2b}(\xi-u)^{2}}L_{A_{1}}^{g}(\xi-u)^{*}e^{-i\frac{ut}{b}+i\frac{d}{2b}u^{2}}e^{i\frac{\xi t}{b}+i\frac{d}{-2b}\xi^{2}}\\
&=\sqrt{-i2\pi b}e^{i\frac{d'}{2b}(\xi-u)^{2}}L_{A_{1}}^{g}(\xi-u)^{*}K_{A}(t,u)K^{*}_{A}(t,\xi)
\end{split}
	\end{align}	
\end{proof}
\noindent{\bf 3. Uncertainty principle for the STLCT}

Uncertainty principle is an exceptionally useful tool for signal processing and mathematics, it is a fundamental result in signal processing \cite{19,22,23,24}. It expresses the relations between time spread and frequency spread. The classic uncertainty principle in signal processing represents that the product of time spread and frequency spread cannot be made arbitrarily small and it has a lower bound. The uncertainty principles of the STFT had been studied in \cite{7,25,26}.
The purpose of this section is to give some new uncertainty principles associated with the STLCT. We first present the following contents.

Let $f(t)$ be a complex signal such that $t|f(t)| \in L^2(\mathbb{R})$ and $E=\int _{-\infty}^{+\infty}\mid f(t)\mid^{2}\rm{d}\textit{t}=\int _{-\infty}^{+\infty}\mid L_{A}^{\textit{f}}(\textit{u})\mid^{2}\rm{d}\textit{u}$ is the energy of the signal. The time and frequency spread are defined by \cite{5}:
\begin{align}
T^{2}_{\textit{f}}=\frac{1}{E}\int _{-\infty}^{+\infty}(t-\overline{t_{f}})^{2}\mid f(t)\mid^{2}\rm{d}\textit{t}
	\end{align}	
\begin{align}
F^{2}_{A,f}=\frac{1}{E}\int _{-\infty}^{+\infty}(u-\overline{u_{f}})^{2}\mid L_{A}^{\textit{f}}(\textit{u})\mid^{2}\rm{d}\textit{u}
	\end{align}	
where $\overline{t_{f}}$ and $\overline{u_{f}}$ are the mean time and mean frequency, respectively. They are defined as \cite{5}:
\begin{align}
\overline{t_{f}}=\frac{1}{E}\int _{-\infty}^{+\infty}t\mid f(t)\mid^{2}\rm{d}\textit{t}
	\end{align}	
\begin{align}
\overline{u_{A,f}}=\frac{1}{E}\int _{-\infty}^{+\infty}u\mid L_{A}^{\textit{f}}(\textit{u})\mid^{2}\rm{d}\textit{u}
	\end{align}	

In order to obtain the uncertainty principles for the STLCT, in a similar way, we define the following parameters ralated to the STLCT. We first defined the time and frequency spread, respectively:
\begin{align}
T^{2}_{A,\textit{S}}=\frac{1}{\|S^{A}_{g}f(t,u)\|^{2}}\int _{-\infty}^{+\infty}\int _{-\infty}^{+\infty}(t-\overline{t_{A,S}})^{2}\mid S^{A}_{g}f(t,u)\mid^{2}\rm{d}\textit{t}\rm{d}\textit{u}
	\end{align}	
\begin{align}
F^{2}_{A,\textit{S}}=\frac{1}{\|S^{A}_{g}f(t,u)\|^{2}}\int _{-\infty}^{+\infty}\int _{-\infty}^{+\infty}(u-\overline{u_{A,S}})^{2}\mid S^{A}_{g}f(t,u)\mid^{2}\rm{d}\textit{t}\rm{d}\textit{u}
	\end{align}	
where $\overline{t_{A,S}}$ and $\overline{u_{A,S}}$ are the mean time and mean frequency, they are defined by:
\begin{align}
\overline{t_{A,S}}=\frac{1}{\|S^{A}_{g}f(t,u)\|^{2}}\int _{-\infty}^{+\infty}\int _{-\infty}^{+\infty}t\mid S^{A}_{g}f(t,u)\mid^{2}\rm{d}\textit{t}\rm{d}\textit{u}
	\end{align}	
\begin{align}
\overline{u_{A,S}}=\frac{1}{\|S^{A}_{g}f(t,u)\|^{2}}\int _{-\infty}^{+\infty}\int _{-\infty}^{+\infty}u\mid S^{A}_{g}f(t,u)\mid^{2}\rm{d}\textit{t}\rm{d}\textit{u}
	\end{align}	
Next, we obtain the following lemma, it is very useful:
\begin{lem} Let $g\in L^2(\mathbb{R})$ be a window function. For a complex signal $f(\tau)=x(\tau)e^{i \phi(\tau)}\in L^{2}(\mathbb{R})$, we have
\begin{align}
	\overline{t_{A,S}}=\overline{t_{f}}+\overline{t_{g}}
	\end{align}	
\begin{align}
	T^{2}_{A,S}=T^{2}_{f}+T^{2}_{g}
	\end{align}	
\begin{align}
	\overline{u_{A,S}}=\overline{u_{A,f}}-\overline{u_{A_{1},g}}
	\end{align}	
\begin{align}
	F^{2}_{A,S}=F^{2}_{A,f}+F^{2}_{A_{1},g}
	\end{align}	
For a window function $g(t)$ and its LCT, where $\overline{t_{g}}$ and $\overline{u_{A_{1},g}}$ are the mean time and mean frequency,
$T^{2}_{g}$ and $F^{2}_{A_{1},g}$ are the time and frequency spread, respectively. Their definitions are similar to  Eq.(18), Eq.(19), Eq.(20) and  Eq.(21), but the matrix that corresponds to $g(t)$ is
$A_{1}=(0,b,-\frac{1}{b},d')$, the parameter $b$ in the matrix $A=(a,b,c,d)$ and $d'\in \mathbb{R}$.
\end{lem}
\begin{proof}
We just derive  Eq.(26) and  Eq.(29) here, and the rest two can be derived in a similar way.

First, using Lemma 1 and  Eq.(4) , we have
\begin{align}
\begin{split}
	\|S^{A}_{g}f(t,u)\|^{2}&=\int _{-\infty}^{+\infty}\int _{-\infty}^{+\infty}|S^{A}_{g}f(t,u)|^{2}\rm{d}\textit{t}\rm{d}\textit{u}\\
&=\int _{-\infty}^{+\infty}\int _{-\infty}^{+\infty}\left[\int _{-\infty}^{+\infty}L_{A}^{f}(\xi_{1})\sqrt{-i2\pi b}e^{i\frac{d'}{2b}(\xi_{1}-u)^{2}}L_{A_{1}}^{g}(\xi_{1}-u)^{*}K_{A}(t,u)K^{*}_{A}(t,\xi_{1})\rm{d}\xi_{1}\right]\\
&\times\left[\int _{-\infty}^{+\infty}L_{A}^{f}(\xi_{2})\sqrt{-i2\pi b}e^{i\frac{d'}{2b}(\xi_{2}-u)^{2}}L_{A_{1}}^{g}(\xi_{2}-u)^{*}K_{A}(t,u)K^{*}_{A}(t,\xi_{2})
\rm{d}\xi_{2}\right]^{*}\rm{d}\textit{t}\rm{d}\textit{u}\\
&=\int _{-\infty}^{+\infty}\int _{-\infty}^{+\infty}|L_{A}^{f}(\xi_{1})|^{2}|L_{A_{1}}^{g}(\xi_{1}-u)|^{2}\rm{d}\xi_{1}\rm{d}\textit{u}\\
\end{split}
	\end{align}	
Let $\xi_{1}=\rho$ and $\xi_{1}-u=\gamma$, then
\begin{align}
\begin{split}
	\|S^{A}_{g}f(t,u)\|^{2}&=
\int _{-\infty}^{+\infty}\int _{-\infty}^{+\infty}|L_{A}^{f}(\rho)|^{2}|L_{A_{1}}^{g}(\gamma)|^{2}\rm{d}\rho\rm{d}\gamma\\
&=\|L_{A}^{f}(\rho)\|^{2}\|L_{A_{1}}^{g}(\gamma)\|^{2}
\end{split}
	\end{align}	
Next, by Definition 2 and  Eq.(4), we have
\begin{align}
\begin{split}
	\overline{t_{A,S}}&=\frac{1}{\|S^{A}_{g}f(t,u)\|^{2}}\int _{-\infty}^{+\infty}\int _{-\infty}^{+\infty}t\mid S^{A}_{g}f(t,u)\mid^{2}\rm{d}\textit{t}\rm{d}\textit{u}\\
&=\frac{1}{\|L_{A}^{f}(\rho)\|^{2}\|L_{A_{1}}^{g}(\gamma)\|^{2}}\int _{-\infty}^{+\infty}\int _{-\infty}^{+\infty}t\left[\int _{-\infty}^{+\infty}f(\xi_{1}) \overline{g(\xi_{1}-t)} K_{A}(\xi_{1},u)\rm{d}\xi_{1}\right]\\
&\times\left[\int _{-\infty}^{+\infty}f(\xi_{2}) \overline{g(\xi_{2}-t)} K_{A}(\xi_{2},u)\rm{d}\xi_{2}\right]^{*}\rm{d}\textit{t}\rm{d}\textit{u}\\
&=\frac{1}{\|L_{A}^{f}(\rho)\|^{2}\|L_{A_{1}}^{g}(\gamma)\|^{2}}\int _{-\infty}^{+\infty}\int _{-\infty}^{+\infty}t|f(\xi_{1})|^{2}|g(\xi_{1}-t)|^{2}\rm{d}\xi_{1}\rm{d}\textit{t} \\
\end{split}
	\end{align}	
Let $\xi_{1}-t=\nu$, we can obtain the mean time of STLCT
\begin{align}
\begin{split}
	\overline{t_{A,S}}&=\frac{1}{\|L_{A}^{f}(\rho)\|^{2}\|L_{A_{1}}^{g}(\gamma)\|^{2}}\int _{-\infty}^{+\infty}\int _{-\infty}^{+\infty}(\xi_{1}-\nu)|f(\xi_{1})|^{2}|g(\nu)|^{2}\rm{d}\xi_{1}\rm{d}\nu \\
&=\frac{1}{\|L_{A}^{f}(\rho)\|^{2}}\int _{-\infty}^{+\infty}\xi_{1}|f(\xi_{1})|^{2}|\rm{d}\xi_{1}-
\frac{1}{\|L_{A_{1}}^{g}(\gamma)\|^{2}}\int _{-\infty}^{+\infty}\nu|g(\nu)|^{2}\rm{d}\nu \\
&=\overline{t_{f}}-\overline{t_{g}}
\end{split}
	\end{align}	
According to the same method,  Eq.(28) can be obtained.

Using  Eq.(31),  Eq.(23) and  Eq.(28), the frequency spread of STLCT is shown
\begin{align}
\begin{split}
	F^{2}_{A,S}&=\frac{1}{\|S^{A}_{g}f(t,u)\|^{2}}\int _{-\infty}^{+\infty}\int _{-\infty}^{+\infty}(u-\overline{u_{A,S}})^{2}\mid S^{A}_{g}f(t,u)\mid^{2}\rm{d}\textit{t}\rm{d}\textit{u}\\
&=\frac{1}{\|L_{A}^{f}(\rho)\|^{2}\|L_{A_{1}}^{g}(\gamma)\|^{2}}\int _{-\infty}^{+\infty}\int _{-\infty}^{+\infty}[(\rho-\gamma)-(\overline{u_{A,f}}-\overline{u_{A_{1},g}})]^{2}
|L_{A}^{f}(\rho)|^{2}|L_{A_{1}}^{g}(\gamma)|^{2}\rm{d}\rho\rm{d}\gamma\\
&=\frac{1}{\|L_{A}^{f}(\rho)\|^{2}}\int _{-\infty}^{+\infty}(\rho-\overline{u_{A,f}})^{2}|L_{A}^{f}(\rho)|^{2}\rm{d}\rho+\frac{1}{\|L_{A_{1}}^{g}(\gamma)\|^{2}}\int _{-\infty}^{+\infty}
(\gamma-\overline{u_{A_{1},g}})|L_{A_{1}}^{g}(\gamma)|^{2}\rm{d}\gamma\\
&-2\frac{1}{\|L_{A}^{f}(\rho)\|^{2}}\int _{-\infty}^{+\infty}(\rho-\overline{u_{A,f}})|L_{A}^{f}(\rho)|^{2}\rm{d}\rho\frac{1}{\|L_{A_{1}}^{g}(\gamma)\|}\int _{-\infty}^{+\infty}
(\gamma-\overline{u_{A_{1},g}})|L_{A_{1}}^{g}(\gamma)|^{2}\rm{d}\gamma\\
&=F^{2}_{A,f}+F^{2}_{A_{1},g}-2(\overline{u_{A,f}}-\overline{u_{A,f}})(\overline{u_{A_{1},g}}-\overline{u_{A_{1},g}})\\
&=F^{2}_{A,f}+F^{2}_{A_{1},g}
\end{split}
	\end{align}	
According to the same method,  Eq.(27) can be obtained.
\end{proof}
The uncertainty principles of the LCT are already known. Adrian Stern \cite{6} generalized the classical uncertainty principles of the LCT for complex signal in time and frequency domains. It can be restated as follow:
Let $f(t) \in L^2(\mathbb{R})$ and $A=(a,b,c,d)$ be a matrix parameter satisfying $a,b,c,d\in \mathbb{R}$ and $ad-bc=1$, then
\begin{align}
	T^{2}_{f}F^{2}_{A,f}\geq\frac{|b|^{2}}{4}
	\end{align}	
where $T^{2}_{f}$ and $F^{2}_{A,f}$ are defined in Eq.(18) and Eq.(19).

Based on the uncertainty principle of the LCT \cite{6}, we obtain the following main result.
\begin{thm} (Uncertainty principle in time and frequency domains)
Let $g\in L^2(\mathbb{R})$ be a window function. $A=(a,b,c,d)$ be a matrix parameter satisfying $a,b,c,d\in \mathbb{R}$ and $ad-bc=1$. For a complex signal $f(\tau)=x(\tau)e^{i \phi(\tau)}\in L^{2}(\mathbb{R})$, then
\begin{align}
	T^{2}_{A,S}F^{2}_{A,S}\geq |b|^{2}
	\end{align}	
\end{thm}	
\begin{proof}
According to Lemma 2, we have $T^{2}_{A,S}F^{2}_{A,S}=(T^{2}_{f}+T^{2}_{g})(F^{2}_{A,f}+F^{2}_{A_{1},g})=T^{2}_{f}F^{2}_{A,f}+T^{2}_{g}F^{2}_{A_{1},g}+T^{2}_{f}F^{2}_{A_{1},g}+T^{2}_{g}F^{2}_{A,f}$,
Using the fact: $n^{2}+m^{2}\geq 2nm$, for $\forall n,m \in \mathbb{R}$, and Eq.(35), we have
\begin{align}
	\begin{split} T^{2}_{A,S}F^{2}_{A,S}&=T^{2}_{f}F^{2}_{A,f}+T^{2}_{g}F^{2}_{A_{1},g}+T^{2}_{f}F^{2}_{A_{1},g}+T^{2}_{g}F^{2}_{A,f}\\
&\geq \frac{|b|^{2}}{4}+\frac{|b|^{2}}{4}+2\sqrt{T^{2}_{f}F^{2}_{A,f}T^{2}_{g}F^{2}_{A_{1},g}}\\
&\geq \frac{|b|^{2}}{2}+\frac{|b|^{2}}{2}=|b|^{2}
\end{split}
	\end{align}	
\end{proof}
When $A=(\cos\alpha,\sin\alpha,-\sin\alpha,\cos\alpha)$, we can obtain the uncertainty principle of the STFrFT for complex signal in time and frequency domains by Theorem 1, which is given in the following corollary.
\begin{corollary}
Let $g\in L^2(\mathbb{R})$ be a window function. For a complex signal $f(\tau)=x(\tau)e^{i \phi(\tau)}\in L^{2}(\mathbb{R})$, then
\begin{align}
	T^{2}_{\alpha,S}F^{2}_{\alpha,S}\geq |\sin\alpha|^{2}
	\end{align}	
\end{corollary}
According to Theorem 1, the uncertainty principle of the STFT for complex signal can be obtained in time and frequency domains when $A=(0,1,-1,0)$, which is given in the following corollary.
\begin{corollary}
Let $g\in L^2(\mathbb{R})$ be a window function. For a complex signal $f(\tau)=x(\tau)e^{i \phi(\tau)}\in L^{2}(\mathbb{R})$, then
\begin{align}
	T^{2}_{STFT}F^{2}_{STFT}\geq 1
	\end{align}	
\end{corollary}
 Zhao \cite{5} derived the uncertainty principles for complex signal in tow LCT domains. It can be restated as follow:
If $f(t) \in L^2(\mathbb{R})$, then
\begin{align}
	F^{2}_{A,f}F^{2}_{B,f}\geq\frac{(a_{1}b_{2}-a_{2}b_{1})^{2}}{4}
	\end{align}	
where $A=(a_{1},b_{1},c_{1},d_{1})$, $B=(a_{2},b_{2},c_{2},d_{2})$, $a_{1}d_{1}-b_{1}c_{1}=1$, $a_{2}d_{2}-b_{2}c_{2}=1$, $F^{2}_{A,f}$ and $F^{2}_{B,f}$ are defined in Eq.(19).

In the following, we will derive the uncertainty principle for complex signal in two STLCT domains.

\begin{thm} (Uncertainty principle in tow STLCT domains)
Let $g\in L^2(\mathbb{R})$ be a window function. For a complex signal $f(\tau)=x(\tau)e^{i \phi(\tau)}\in L^{2}(\mathbb{R})$, then
\begin{align}
	F^{2}_{A,S}F^{2}_{B,S}\geq \frac{(a_{1}b_{2}-a_{2}b_{1})^{2}}{4}
	\end{align}	
where $A=(a_{1},b_{1},c_{1},d_{1})$, $B=(a_{2},b_{2},c_{2},d_{2})$, $a_{1}d_{1}-b_{1}c_{1}=1$, $a_{2}d_{2}-b_{2}c_{2}=1$.
\end{thm}	
\begin{proof}
The proof of Theorem 2 is similar to Theorem 1.
\end{proof}
When $A=(\cos\alpha,\sin\alpha,-\sin\alpha,\cos\alpha)$, $B=(\cos\beta,\sin\beta,-\sin\beta,\cos\beta)$, we can obtain the uncertainty principle for complex signal in two STLCT domains by Theorem 2, which is given in the following corollary.
\begin{corollary}
Let $g\in L^2(\mathbb{R})$ be a window function. For a complex signal $f(\tau)=x(\tau)e^{i \phi(\tau)}\in L^{2}(\mathbb{R})$, then
\begin{align}
	F^{2}_{\alpha,S}F^{2}_{\beta,S}\geq \frac{\sin^{2}(\alpha-\beta)}{4}
	\end{align}	
\end{corollary}

In order to obtain an uncertainty relation for the two conditional standard deviations of the spectrogram, we first
give the mean averages and conditional standard deviations, as follow:
\begin{align}
	\langle u\rangle_{t}=\frac{1}{Q(t)}\int _{-\infty}^{+\infty}u|S^{A}_{g}f(t,u)|^{2}\rm{d}\textit{u}
	\end{align}	
\begin{align}
	\langle t\rangle_{u}=\frac{1}{P(u)}\int _{-\infty}^{+\infty}t|s^{A}_{g}f(t,u)|^{2}\rm{d}\textit{t}
	\end{align}	
\begin{align}
	\sigma^{2}_{u|t}=\frac{1}{Q(t)}\int _{-\infty}^{+\infty}(u-\langle u\rangle_{t})^{2} |S^{A}_{g}f(t,u)|^{2}\rm{d}\textit{u}
	\end{align}	
\begin{align}
	\sigma^{2}_{t|u}=\frac{1}{P(u)}\int _{-\infty}^{+\infty}(t-\langle t\rangle_{u})^{2} |s^{A}_{g}f(t,u)|^{2}\rm{d}\textit{t}
	\end{align}	
\begin{thm} (Uncertainty principle for the two conditional standard deviations of the spectrogram)
Let $g\in L^2(\mathbb{R})$ be a window function. For a complex signal $f(\tau)=x(\tau)e^{i \phi(\tau)}\in L^{2}(\mathbb{R})$, then
\begin{align}
	\sigma^{2}_{u|t}\sigma^{2}_{t|u}\geq\frac{1}{Q(t)P(u)}\left|\int _{-\infty}^{+\infty}(S^{A}_{g}f(t,u))^{*}
(\frac{1}{2}[\mathcal{A},\mathcal{B}]_{+}+\frac{i}{2})
\textit{f}_{t}(\tau)\rm{d}\tau\right|^{2}
	\end{align}	
where $\mathcal{A}, \mathcal{B}$ are the Hermitian operators.
\end{thm}	
\begin{proof}
According to Eq.(45), we have
\begin{align}
\begin{split}
	\sigma^{2}_{u|t}&=\frac{1}{Q(t)}\int _{-\infty}^{+\infty}(u-\langle u\rangle_{t})^{2} S^{A}_{g}f(t,u)(S^{A}_{g}f(t,u))^{*}\rm{d}\textit{u}\\
&=\frac{1}{Q(t)}\int _{-\infty}^{+\infty}(u-\langle u\rangle_{t})^{2}
\int _{-\infty}^{+\infty}f(\tau)g^{*}(\tau-t)K_{A}(\tau,u)\rm{d}\tau
\int _{-\infty}^{+\infty}f(\tau')g^{*}(\tau'-t)K_{A}(\tau',u)\rm{d}\tau'\rm{d}\textit{u}\\
&=\frac{1}{Q(t)}\int _{-\infty}^{+\infty}(\frac{b}{i}\frac{d}{d\tau}-\langle u\rangle_{t})^{2}
f(\tau)g^{*}(\tau-t)(f(\tau)g^{*}(\tau-t))^{*}\rm{d}\tau\\
&=\frac{1}{Q(t)}\int _{-\infty}^{+\infty}\left|(\frac{b}{i}\frac{d}{d\tau}-\langle u\rangle_{t}) f_{t}(\tau)\right|^{2}\rm{d}\tau\\
\end{split}
	\end{align}	
Let
\begin{align}
	\mathcal{A}=\tau-\langle t\rangle_{u};\ \ \mathcal{B}=\frac{b}{i}\frac{d}{d\tau}-\langle u\rangle_{t}
	\end{align}
then $\mathcal{A}, \mathcal{B}$ are the Hermitian operators [8,9,10].

So we have
\begin{align}
	\mathcal{A}\mathcal{B}=\frac{1}{2}[\mathcal{A},\mathcal{B}]_{+}+\frac{i}{2}([\mathcal{A},\mathcal{B}]/i)
	\end{align}
and
\begin{align}
	[\mathcal{A},\mathcal{B}]=i
	\end{align}
By Eq.(46), Eq.(48) and Eq.(49), we obtain
\begin{align}
	\sigma^{2}_{t|u}\sigma^{2}_{u|t}=\frac{1}{P(u)}\int _{-\infty}^{+\infty}\mathcal{A}^{2} |S^{A}_{g}f(t,u)|^{2}\rm{d}\tau\frac{1}{Q(t)}\int _{-\infty}^{+\infty}\mathcal{B}^{2} |\textit{f}_{t}(\tau)|^{2}\rm{d}\tau
	\end{align}
According to the Schwarz inequality
\begin{align}
	\int _{-\infty}^{+\infty}|f(t)|^{2}\rm{d}\textit{t}\int _{-\infty}^{+\infty}|\textit{g}(t)|^{2}\rm{d}\textit{t}
\geq\left|\int _{-\infty}^{+\infty}\textit{f}^{*}(t)\textit{g}(t)\rm{d}\textit{t}\right|^{2}
	\end{align}
we can obtain
\begin{align}
\begin{split}
	\sigma^{2}_{t|u}\sigma^{2}_{u|t}&=\frac{1}{Q(t)P(u)}\int _{-\infty}^{+\infty}\mathcal{A}^{2} |S^{A}_{g}f(t,u)|^{2}\rm{d}\tau\int _{-\infty}^{+\infty}\mathcal{B}^{2} |\textit{f}_{t}(\tau)|^{2}\rm{d}\tau\\
&\geq\frac{1}{Q(t)P(u)}\left|\int _{-\infty}^{+\infty}(\mathcal{A}S^{A}_{g}f(t,u))^{*}\mathcal{B}\textit{f}_{t}(\tau)\rm{d}\tau\right|^{2}\\
&=\frac{1}{Q(t)P(u)}\left|\int _{-\infty}^{+\infty}(S^{A}_{g}f(t,u))^{*}\mathcal{A}\mathcal{B}\textit{f}_{t}(\tau)\rm{d}\tau\right|^{2}
\end{split}
	\end{align}
Using Eq.(50) and Eq.(51), we have
\begin{align}
\begin{split}
	\sigma^{2}_{t|u}\sigma^{2}_{u|t}
&\geq\frac{1}{Q(t)P(u)}\left|\int _{-\infty}^{+\infty}(S^{A}_{g}f(t,u))^{*}
(\frac{1}{2}[\mathcal{A},\mathcal{B}]_{+}+\frac{i}{2}([\mathcal{A},\mathcal{B}]/i))
\textit{f}_{t}(\tau)\rm{d}\tau\right|^{2}\\
&=\frac{1}{Q(t)P(u)}\left|\int _{-\infty}^{+\infty}(S^{A}_{g}f(t,u))^{*}
(\frac{1}{2}[\mathcal{A},\mathcal{B}]_{+}+\frac{i}{2})
\textit{f}_{t}(\tau)\rm{d}\tau\right|^{2}\\
\end{split}
	\end{align}
\end{proof}

\noindent  {\bf 4. Conclusions}

In this paper, many different forms of uncertainty principles for complex signals associated with the STLCT are derived.
Firstly, based on the relation between the STLCT and the LCT, one uncertainty principle for the STLCT of complex signals in time and frequency domains is derived. Secondly, the uncertainty
principle for the STLCT of complex signals in tow STLCT domains is obtained. Finally, according to the Hermitian operators, the uncertainty principle for the two conditional standard deviations of the spectrogram associated with the STLCT is discussed.
These uncertainty principles are new results
and it can be viewed as a generalized form of the STFT.

Our further work will think about these applications of the STLCT in signal processing and these uncertainty principles for discrete signals.

\medskip

\noindent  {\bf Acknowledge}

This work is supported by the National Natural Science Foundation of China (No. 61671063).



\begin{thebibliography}{99}


\bibitem{1} S. C. Pei, J. J. Ding, Eigenfunctions of linear canonical transform. IEEE Trans. Signal Processing. {\bf50}(1), 11-26 (2002)
\bibitem{2} B. Z. Li, R. Tao, Y. Wang, New sampling formulae related to linear canonical transform. Signal Process. {\bf87}(5), 983-990 (2007)
\bibitem{3} T. Z. Xu, B. Z. Li, {\it Linear Canonical Transform and Its Applications}. Science Press, Beijing, China (2013)
\bibitem{4} K. I. Kou, R. H. Xu, Windowed linear canonical transform and its applications. Signal Process. {\bf92}(1), 179-188 (2012)
\bibitem{5} J. Zhao, R. Tao, Y. L. Li, Uncertainty principles for linear canonical transform. IEEE Trans. Signal Process. {\bf57}(7), 2856-2858 (2009)
\bibitem{6} A. Stern, Uncertainty Principles in Linear Canonical Transform Domains and Some of Their Implications in Optics. J. Opt. Soc. Amer. A. {\bf25}(3), 647-652 (2008)
\bibitem{7} L. Cohen, Uncertainty principles of the short-time Fourier transform. Advanced Signal Processing Algorithms. Proc. SPIE. {\bf2563}, 80-90 (1995)
\bibitem{8} J. Shi, X. Liu, N. Zhang, On uncertainty principles for linear canonical transform of complex signals via operator methods. Signal Image Video Process. {\bf8}(1), 85-93 (2014)
\bibitem{9} G. B. Folland, A. Sitaram, The uncertainty principle: a mathematical survey. J. Fourier Anal. Appl. {\bf3}(3), 207-238 (1997)
\bibitem{10} L. Cohen, {\it Time-Frequency Analysis}. 383-399, Prentice-Hall (1995)
\bibitem{11} Q. Y. Zhang, Discrete windowed linear canonical transform. ICSPCC. Hong Kong, China (2016)
\bibitem{12} Z. C. Zhang, Sampling theorem for the short-time linear canonical transform and its applications. Signal Process. {\bf113}, 138-146 (2015)
\bibitem{13} M. F. Erden, M. A. Kutay, H. M. Ozaktas, Repeated filtering in consecutive fractional Fourier domains and its application to signal restoration. IEEE Trans. Signal Process. {\bf47}(5), 1458-1462 (1999)
\bibitem{14} Y. X. Fu, L. Q. Li, Generalized analytic signal associated with linear canonical transform. Opt. Commun. {\bf281}, 1468-1472 (2008)
\bibitem{15} H. Y. Huo, Uncertainty Principles for the Offset Linear Canonical Transform. Circuits Syst. Signal Process. {\bf38}(1), 395-406 (2019)
\bibitem{16} Z. C. Zhang, Unified Wigner-Ville distribution and ambiguity function in the linear canonical transform domain. Signal Process. {\bf114}, 45-60 (2015)
\bibitem{17} M. Bahri, R. Ashino, Some properties of windowed linear canonical transform and its logarithmic uncertainty principle. Int. J. Wavelets Multiresolut Inf. Process. {\bf14}(3), 1650015 (2016)
\bibitem{18} K. I. Kou, R. H. Xu, Y. H. Zhang, Paley-Wiener theorems and uncertainty principles for the windowed linear canonical transform. Math. Methods Appl. Sci. {\bf35}(17), 2122-2132 (2012)
\bibitem{19} J. Shi, X. Liu, N. Zhang, On uncertainty principles of linear canonical transform for complex signals via operator methods. Signal Image Video Process. {\bf8}(1), 85-93 (2014)
\bibitem{20} R. Tao, B. Z. Li, Y. Wang, et al. On sampling of band-limited signals associated with the linear canonical transform. IEEE Trans. Signal Process. {\bf56}(11), 5454-5464 (2008)
\bibitem{21}S. Q. Xu, L. Feng, Y. Chai, Y. G. He, Analysis of A-stationary random signals in the linear canonical transform domain. Signal Process. {\bf146}, 126-132 (2008)
\bibitem{22} K. K. Sharma, S. D. Joshi, Uncertainty principles for real signals in linear canonical transform domains. IEEE Trans. Signal Process. {\bf56}(7), 2677-2683 (2008)
\bibitem{23} S. Shinde and V. M. Gadre, An uncertainty principle for real signals in the fractional Fourier transform domain. IEEE Trans. Signal Process. {\bf49}(11), 2545-2548 (2001)
\bibitem{24} T. Li, J. F. Zhang, Sampled-data based average consensus with measurement noises: convergence analysis and uncertainty principle. Sci. China Inf. Sci. {\bf52}(11), 2089-2103 (2009)
\bibitem{25} A. Bonami, B. Demange, P. Jaming, Hermite functions and uncertainty principles for the Fourier and the windowed Fourier transforms. Revista Matem$\acute{a}$tica Iberoamericana. {\bf19}(1), 23-55 (2003)
\bibitem{26} K. Gr\"{o}chenig, Uncertainty principles for time-frequency representations. Advances in Gabor analysis. Birkh\"{a}user, Boston, MA (2003)
\end{thebibliography}
\end{document}